\documentclass[conference,letterpaper]{IEEEtran}

\usepackage{amsthm,amsfonts,braket,algorithm,graphicx}
\usepackage{balance}

\usepackage[utf8]{inputenc} 
\usepackage[T1]{fontenc}
\usepackage{url}
\usepackage{ifthen}
\usepackage[cmex10]{amsmath} 

\theoremstyle{definition} 
\theoremstyle{definition} 
\newtheorem {theorem} {Theorem}

\newcommand{\kb}[1]{\mathbf{\left[#1\right]}}

\newcommand{\trd}[1]{\left|\left| #1 \right| \right|}

\newcommand{\st}{\text{ } : \text{ }}

\newcommand{\Hmin}{H_\infty}
\newcommand{\Hmax}{H_{max}}
\newcommand{\Hextd}{\bar{H}}

\newcommand{\hw}{w}

\newcommand{\alphabet}{\mathcal{A}}

\newcommand{\W}{\mathcal{W}}
\newcommand{\Z}{\mathcal{Z}}

\floatname{algorithm}{Protocol}

\begin{document}
\title{Semi-Source Independent Quantum Walk Random Number Generation}

\author{%
    \IEEEauthorblockN{Minwoo Bae and Walter O. Krawec}
  \IEEEauthorblockA{University of Connecticut\\
                    Department of Computer Science and Engineering\\
                    Storrs, CT, USA 06269\\
                    Email: walter.krawec@uconn.edu}
}

\maketitle

\begin{abstract}
Semi-source independent quantum random number generators (SI-QRNG) are cryptographic protocols which attempt to extract random strings from quantum sources where the source is under the control of an adversary (but with known dimension) while the measurement devices are fully characterized.  This represents a middle-ground between fully-trusted and full-device independence, allowing for fast bit-generation rates with current-day technology, while also providing a strong security guarantee.  In this paper we analyze a SI-QRNG protocol based on quantum walks and develop a proof of security.  We derive a novel entropic uncertainty relation for this application which is necessary since standard relations actually fail in this case.
\end{abstract}

\section{Introduction}

Random number generation is an important process for a variety of application domains.  Due to the intrinsic randomness of quantum processes, quantum random number generation (QRNG) is an important field of study within quantum information science.
By now, cryptographically secure QRNG protocols are well studied under a variety of security models ranging from the ``fully trusted device'' scenario (whereby all devices used, sources and measurements, are fully characterized) to the ``fully device independent'' scenario (where all devices used are not trusted) \cite{di-qrng1,di-qrng2}.  Clearly from a cryptographic point of view, DI-QRNG protocols are the desirable ideal due to their minimal assumptions needed for security.  However, though experimental progress has been rapidly improving, the bit-rates of such protocols cannot compare to other models \cite{di-qrng-exp,di-qrng-exp-2}.  As a compromise, the \emph{source independent} (SI) model was introduced in \cite{vallone2014quantum} whereby measurement devices are characterized (though not necessarily ideal) whereas the source is under the control of the adversary.  One may envision the source being a quantum server, providing a service to users who wish to distill cryptographically secure random strings without trusting the server (e.g., the server may be adversarial).  The SI model affords fast experimental bit generation rates \cite{si-qrng-fast} (with a recent paper discussing an implementation with a rate over 8Gb/s \cite{si-qrng-fast-new}) along with fascinating potential applications, including the use of sunlight as a source \cite{si-qrng-sun}.  For a survey of QRNG protocols, the reader is referred to \cite{qrng-survey}.  Note that we are actually considering a semi-source independent model where the dimension of the source is known but no other assumptions are made (this is exactly the model introduced in \cite{vallone2014quantum}).

Outside of QRNG's, quantum walks (QW), the quantum analogue of classical random walks, are a highly important process in quantum computation \cite{farhi1998quantum,childs2003exponential,childs2009universal,lovett2010universal} and, recently, in quantum cryptography \cite{rohde2012quantum,vlachou2015quantum,vlachou2018quantum,srikara2020quantum}.  Recently, a QW-based random number generation protocol was analyzed in \cite{QW-QRNG}, though a rigorous security analysis was not done.  In this paper, we revisit that protocol, minimally changing it to be a SI-QRNG protocol, and prove its security.  To our knowledge, this is the first SI-QRNG protocol with provable composable security based on quantum walks.  We note that the security analysis of this protocol is not trivial.  Due to certain simplifications we make to allow for an easier potential experimental implementation, prior tools are not immediately applicable (though we do not consider experimental concerns in this work, we keep them in mind when developing the protocol). In this work we develop an alternative entropic uncertainty relation which may also hold applications in other quantum cryptographic protocols.

Naturally, QW's are random processes and, so, at first glance designing and proving secure, a QW-QRNG protocol seems a trivial task.  Indeed, the following protocol is a trivial solution to the problem with an ``easy'' (using modern information theoretic tools) security proof in the SI model:
(1) First, a source prepares a state $\ket{\psi_{0,0}}^{\otimes (n+m)}$ where $\ket{\psi_{0,0}}$ is some quantum walker state.  While we discuss this in detail later, for now it suffices to consider $\ket{\psi_{0,0}} = W\ket{0,0}$ where $W$ is a unitary operator and $\ket{0,0}$ lives in some Hilbert space of dimension $2P$.  This state is sent to Alice.
(2) Second, Alice chooses a random subset of size $m$ and measures the systems indexed by this subset in the ``quantum walk basis'', namely the orthonormal basis $\{W\ket{0,0}, W\ket{0,1},\cdots, W\ket{1,P-1}\}$.  Ideally, this measurement should always produce the zeroth state of this basis.  The remaining $n$ walker systems are measured in the computational basis $\{\ket{0,0}, \ket{0,1}, \cdots, \ket{1,P-1}\}$.  The first outcome is used to test the fidelity of the received state while the second is used as a raw-random string.  This string is then further processed through a privacy amplification process, the output of which is the final cryptographic random string.

Indeed this protocol can be proven secure in a very straight-forward manner using entropic uncertainty \cite{berta2010uncertainty,tomamichel2011uncertainty,ent-survey}.  However, there are two complications with the protocol itself.  First, it would require the ability for Alice to perform a full basis measurement in the quantum-walk basis (namely, she would need to distinguish all states of the form $W\ket{c,x}$).  This might require complex optics to do experimentally.  Second, for the randomness generation measurement, she needs to be able to perform a measurement in the full coin and position basis, namely a measurement that can distinguish all states of the form $\ket{c,x}$.  Our goal is to analyze a far simpler protocol, building off of the one from \cite{QW-QRNG}.  The protocol will only require Alice to be able to distinguish a single walker state, namely $W\ket{0,0}$ from any other; and, second, she need only perform a measurement of the position of the walk for randomness, and she need not also determine the state of the coin itself.  The second restriction is identical to the protocol in \cite{QW-QRNG} though, since they did not consider the source independent model, they did not require any other test.  We add only this minimal test ability, namely the ability to distinguish a single quantum walk state from the $2P-1$ others in the walk basis, to ensure a cryptographically secure protocol.

Interestingly, standard entropic uncertainty relations of the form \cite{tomamichel2011uncertainty}:
\begin{equation}\label{eq:ent}
\Hmin^\epsilon(A|E) + \Hmax^\epsilon(A|B) \ge -\log\max_{x,y}||\sqrt{M_x}\sqrt{N_y}||^2_{op},
\end{equation}
where $\{M_x\}$ and $\{N_y\}$ are the two POVMs used in the protocol, are not applicable and can only yield the trivial bound.  Thus a new approach is required to analyze this QW-QRNG protocol.  We develop the approach in this paper using a technique of quantum sampling as introduced by Bouman and Fehr in \cite{sampling} and used by us recently to develop novel \emph{sampling-based entropic uncertainty relations} \cite{krawec2019quantum,krawec2020new}.  In fact, our proof is similar, though with suitable modifications needed for this scenario and, since the result does not follow immediately from our previous analysis, it is necessary to state here.  

We make two primary contributions in this paper.  First, we analyze for the first time, a QW-QRNG protocol introduced in \cite{QW-QRNG} from a cryptographic perspective.  We adapt the protocol sufficiently, and minimally, so as to produce a secure system and prove it is secure in the SI model.  This represents, to our knowledge, the first QRNG protocol based on quantum walks in the SI model of security and shows even greater application of quantum walks to other cryptographic primitives.  Second, we develop a proof of security to handle this scenario when standard approaches are not immediately applicable.  Our security method may also be applicable to other protocols of this nature where standard relations such as Equation \ref{eq:ent} cannot be used directly.  Our proof utilizes the method of quantum sampling by Bouman and Fehr \cite{sampling}, augmented with techniques we developed in \cite{krawec2019quantum,krawec2020new} for entropic uncertainty, showing even more potential applications of these methods to complex security analyses.  We actually think this second contribution the more significant as it shows how this framework of quantum sampling may be used to tackle cryptographic problems that standard methods would fail to analyze successfully, thus opening the door to a potential wider range of applications.

\section{Notation and Definitions}

We now introduce some basic definitions and notation that we will use throughout this paper.  By $\alphabet_d$ we mean a $d$-dimensional alphabet, namely $\alphabet_d = \{0, 1, \cdots, d-1\}$.  Given a word $q \in \alphabet_d^N$ and some subset $t \subset \{1, 2, \cdots, N\}$, we write $q_t$ to mean the substring of $q$ indexed by $t$ (i.e., those characters in $q$ indexed by $i \in t$).  We write $q_{-t}$ to mean the substring indexed by the complement of $t$.  The \emph{Hamming Weight} of $q$ is denoted $wt(q) = |\{i \st q_i \ne 0\}|$ while the \emph{relative Hamming weight} is denoted $w(q) = wt(q)/|q|$.

A \emph{density operator} is a Hermitian positive semi-definite operator of unit trace acting on some Hilbert space $\mathcal{H}$.  Given a pure quantum state $\ket{\psi} \in \mathcal{H}$ we write $\kb{\psi}$ to mean $\ket{\psi}\bra{\psi}$.

The Shannon entropy of a random variable $X$ is denoted by $H(X)$ while the $d$-ary entropy function is denoted $h_d(x)$.  This function is defined to be $h_d(x) = x\log_d(d-1) - x\log_dx - (1-x)\log_d(1-x)$.  We also define the \emph{extended $d$-ary entropy function} to be $\Hextd_d(x)$ which equals $h_d(x)$ for all $x \in [0, 1-1/d]$ but is $0$ for all $x < 0$ and is $1$ for all $x > 1-1/d$.

Let $\rho_{AE}$ be a quantum state (density operator) acting on some Hilbert space $\mathcal{H}_A\otimes\mathcal{H}_E$.  The \emph{conditional quantum min entropy} \cite{renner2008security} is defined to be:
$\Hmin(A|E)_\rho = \sup_{\sigma_E}\max(\lambda\in\mathbb{R} \st 2^{-\lambda}I_A\otimes\sigma_E - \rho_{AE} \ge 0),$
where $I_A$ is the identity operator on $\mathcal{H}_A$.  Note that if the $E$ system is trivial and the $A$ portion is classical (namely $\rho_A = \sum_xp_x\kb{x}$) then it is easy to show that $\Hmin(A) = -\log \max_xp_x$.  If the $E$ portion is classical, namely $\rho_{AE} = \sum_ep_e\rho_A^e\otimes\kb{e}$, then it can be shown that:
\begin{equation}\label{eq:cl-ent}
\Hmin(A|E)_\rho \ge \min_e\Hmin(A)_{\rho^e}.
\end{equation}
Finally, the \emph{smooth conditional min entropy} is defined to be \cite{renner2008security}:
$\Hmin^\epsilon(A|E)_\rho = \sup_{\sigma\in\Gamma_\epsilon(\rho)}\Hmin(A|E)_\sigma,$
with:
$\Gamma_\epsilon(\rho) = \{\sigma \st \trd{\sigma-\rho} \le \epsilon\}.$
Here, $\trd{X}$ is the trace distance of operator $X$.


Given a classical-quantum state $\rho_{AE}$, let $\sigma_{KE}$ be the result of a \emph{privacy amplification} process on the $A$ register of this state.  Namely, a process of mapping the $A$ register through a randomly chosen two-universal hash function.  If the output of this hash function is $\ell$ bits long, then it was shown in \cite{renner2008security} that:
\begin{equation}\label{eq:PA}
\trd{\sigma_{KE} - I_K/2^\ell\otimes\sigma_E} \le 2^{-\frac{1}{2}(\Hmin^\epsilon(A|E)_\rho - \ell)} + 2\epsilon.
\end{equation}

\subsection{Quantum Random Walks}

In this work we will consider discrete-time quantum walks on a cycle graph \cite{aharonov2001quantum}.  Such a process involves a Hilbert space $\mathcal{H}_W = \mathcal{H}_C\otimes\mathcal{H}_P$ where $\mathcal{H}_C$ is the two-dimensional \emph{coin space} and $\mathcal{H}_P$ is the $P$-dimensional \emph{position space}.  The walk begins with the walker at some initial position $\ket{c,x}$ (e.g., $\ket{0,0}$) from which a \emph{walk operator} is applied $T$ times.  The walk operator first applies a unitary operator on the coin space (for us, we only consider the Hadamard operator here, though other possibilities exist of course).  Following this a \emph{shift operator} is applied $S$ which maps $\ket{0,x}\mapsto\ket{0,x+1}$ and $\ket{1,x}\mapsto\ket{1,x-1}$ where all arithmetic in the position space is done modulo $P$.  Let $W = S\cdot (H\otimes I_P)$ be the walk operator; then, after $T$ steps, the walker evolves to state $W^T\ket{c,x}$.  Generally, at this point, a measurement may be done on the position space causing a collapse at one of the $P$ spots.

Later, we will denote by $\ket{w_{c,x}}$ to mean the evolved state $W^T\ket{c,x}$.  We will also use $\ket{w_i}$ when appropriate, using the natural relationship of tuples $(c,x)$ to integers $i$, with $(0,0)$ being the first index $i=0$.  Finally, given a walk state $\ket{w_{c,x}}$ we use the notation $Pr_W(\ket{w_{c,x}}\rightarrow z)$ to denote the probability that the walker is observed at position $z$ after measurement.  Namely, $Pr_W(\ket{w_{c,x}}\rightarrow z) = \braket{w_{c,x}|I_C\otimes\kb{z}|w_{c,x}}.$
Finally, we denote by $\gamma$ to be the maximal positional probability of the walk which starts at $\ket{0,0}$, namely:
\begin{equation}\label{eq:gamma}
\gamma = \max_zPr_W(\ket{w_{0,0}}\rightarrow z).
\end{equation}
Obviously, this is a function of the walk parameters (the operation $W$ along with the number of steps $T$).

\subsection{Quantum Sampling}

In \cite{sampling}, Bouman and Fehr discovered a fascinating connection linking classical sampling strategies with quantum ones, even when the quantum state is entangled with an environment system (e.g., an adversary).  Here we review some of these concepts, however for more details the reader is referred to \cite{sampling}.

Let $q\in\alphabet_d^N$.  A classical sampling strategy is a process of choosing a random subset $t \subset\{1, \cdots, N\}$, observing $q_t$, and estimating the value of some target value of the \emph{unobserved portion}.  Here, as in \cite{sampling}, we consider the target value to be the relative Hamming weight.  One sampling strategy we will employ consists of choosing a subset $t$ of size $m \le N/2$ uniformly at random, observing $q_t$, and outputting $w(q_t)$ as an estimate of the Hamming weight in the unobserved portion.  It was shown in \cite{sampling} that, for $\delta > 0$:
\begin{equation}\label{eq:err-cl}
\epsilon_\delta^{cl} := \max_{q\in\mathcal{A}_d^N}Pr(q \not\in \mathcal{G}_{t,\delta}) \le 2\exp\left(\frac{-\delta^2m(n+m)}{m+n+2}\right),
\end{equation}
where the probability is over all choices of subsets $t$ and $\mathcal{G}_{t,\delta}$ is the set of all ``good'' words for which this sampling strategy is guaranteed to produce a $\delta$-close estimate of the Hamming weight of the unobserved portion, namely:
\[
\mathcal{G}_{t,\delta} = \{q\in\alphabet_{d}^{N} \st |w(q_t) - w(q_{-t})|\le\delta\}.
\]
The value $\epsilon_\delta^{cl}$ is the error probability of the classical sampling strategy (the ``cl'' superscript is used to refer to a classical sampling strategy).

The main result from \cite{sampling} shows how to promote such a classical strategy to a quantum one in a way that the failure probabilities of the quantum strategy are functions of the classical ones.  Fix a basis $\{\ket{0}, \cdots, \ket{d-1}\}$ (the exact choice may be arbitrary but then fixed - later when using this result, we will use the walk basis $\{W^T\ket{c,x}\}_{c,x}$).  Define:
\[
span(\mathcal{G}_{t,\delta}) = span(\ket{i_1i_2\cdots i_N} \st |w(i_t) - w(i_{-t})|\le\delta).
\]
This is the quantum analogue of the ``good set'' of classical words.  In particular, note that if given a state $\ket{\phi}_{AE} \in span(\mathcal{G}_{t,\delta})\otimes \mathcal{H}_E$, then if a measurement in the given basis were performed on those qudits indexed by $t$ leading to outcome $q\in\mathcal{A}_d^m$, it must hold that the remaining state is a superposition of the form:
$\ket{\phi_{t,q}} = \sum_{i\in J}\alpha_i\ket{i,E_i},$
where $J\subset \{i \in \mathcal{A}_d^{N-m} \st |w(i) - w(q)| \le \delta\}$.

The main result from \cite{sampling}, reworded for our application here, was to prove the following theorem:
\begin{theorem}\label{thm:sample}
(Modified from \cite{sampling}): Let $\delta > 0$.  Given the above classical sampling strategy and an arbitrary quantum state $\ket{\psi}_{AE}$, there exists a collection of ``ideal states'' $\{\ket{\phi^t}\}_{t}$, indexed over all possible subsets the sampling strategy may choose, such that each $\ket{\phi^t} \in span(\mathcal{G}_{t,\delta})\otimes\mathcal{H}_E$ and:
\begin{equation}\label{eq:ideal}
\frac{1}{2}\trd{\frac{1}{T}\sum_t\kb{t}\otimes\kb{\psi} - \frac{1}{T}\sum_t\kb{t}\otimes\kb{\phi^t}} \le \sqrt{\epsilon_\delta^{cl}}.
\end{equation}
where $T = {N \choose m}$ and the sum is over all subsets of size $m$.
\end{theorem}
Note that the result requires a fixed basis of reference (from which to define $\mathcal{G}_{t,\delta}$).

\section{The Protocol}

We consider a QW-QRNG protocol introduced in \cite{QW-QRNG}.  That protocol was not analyzed rigorously from a cryptographic standpoint and, in fact, would not be secure in the SI model.  We modify that protocol, adding a minimal testing ability for Alice, and later show it is secure in the SI model of security.  The protocol operates as follows:
$ $\newline\textbf{Public Parameters:} The quantum walk setting, namely the dimension of the position space $P$ (defining the overall Hilbert space of one walker $\mathcal{H}_W = \mathcal{H}_C\otimes\mathcal{H}_P$), the walk operator $W$, and the number of steps to evolve by, $T$.
$ $\newline\textbf{Source:} A source, potentially adversarial, produces a quantum state $\ket{\psi_0} \in \mathcal{H}_A\otimes\mathcal{H}_E$, where $\mathcal{H}_A \cong \mathcal{H}_W^{\otimes N}$.  If the source is honest, the state prepared should be of the form:
  \[
  \ket{\psi_0} = \ket{w_0}^{\otimes N}\otimes\ket{0}_E,
  \]
  namely, $N$ copies of the walker state $\ket{w_0} = \ket{w_{0,0}} = W^T\ket{0,0}$ unentangled with Eve.
$ $\newline\textbf{User:} Alice chooses a random subset $t$ of size $m$ and measures those walker states using POVM $\W = \{\kb{w_0}, I-\kb{w_0}\} = \{W_0, W_1\}$ resulting in outcome $q \in \{0,1\}^m$ (equivalently, she reverses the quantum walk and observes whether the initial state was $\ket{0,0}$ or anything else).  The remaining states she measures using POVM $\Z = \{I_C\otimes\kb{j}\}_{j=0}^{P-1} = \{Z_0, Z_1, \cdots, Z_{P-1}\}$ resulting in outcome $r \in \alphabet_P^n$, where $n = N-m$.
$ $\newline\textbf{Postprocessing:} Finally, Alice applies privacy amplification to $r$, producing a final random string of size $\ell$.  As proven in \cite{frauchiger2013true}, the hash function used for privacy amplification need only be chosen randomly once and then reused for each run of the protocol for a QRNG protocol of this nature.

The goal of this protocol is to ensure that, for a given $\epsilon_{PA}$ set by the user, after privacy amplification the resulting string is $\epsilon_{PA}$ close to an ideal random string, uniformly generated and independent of any adversary system.  Using Equation \ref{eq:PA}, this involves finding a bound on the quantum min-entropy.  Note that, for the given POVMs, it is straight-forward to check that $\max_{x,y}||\sqrt{W_x}\sqrt{Z_y}||^2_{op} = 1$ and so Equation \ref{eq:ent} only yields the trivial bound on the min entropy.  Thus an alternative approach is required which we develop in the next section.

\subsection{Security Analysis}

To prove security, we require a bound on the quantum min entropy from which, using Equation \ref{eq:PA}, we may compute the number of random bits $\ell$ which may be extracted from $N$ quantum walk states (prepared by an adversary).  We assume the adversary is allowed to create any initial state, possibly entangled with her ancilla, however as in \cite{vallone2014quantum}, the dimension of the system sent to Alice is known; in our case it is $(2P)^N$, namely, $N$ quantum walker states, each of dimension $2P$.  We do not assume anything else about this state (for instance, each of the $N$ walkers may be in different states).  Such a scenario also models natural noise and an honest source - considering an adversarial source is more general.  Finally, we assume that Alice's measurement devices are fully characterized.

\begin{theorem}\label{thm:main}
Let $\epsilon > 0$.  After executing the above QW-QRNG protocol and observing outcome $q$ during the test stage (namely, after measuring using $\W$), it holds that, except with probability at most $\epsilon^{1/3}$ (where the probability here is over the choice of sample subset and observation $q$), the protocol outputs a final secret string of size:
\[
\ell=-\eta_q\log_2\gamma - n\cdot\frac{\Hextd_{2P}(w(q) + \delta)}{\log_{2P} (2)} - 2\log_2\frac{1}{\epsilon} - \log_2{N \choose m},
\]
which is $(5\epsilon + 4\epsilon^{1/3})$-close to an ideal random string (i.e., one that is uniformly generated and independent of any adversary system as in Equation \ref{eq:PA}).  Above, $\eta_q = (N-m)(1-w(q)-\delta)$ and:
\begin{equation}\label{eq:delta}
  \delta = \sqrt{\frac{(N+2)\ln(2/\epsilon^2)}{m\cdot N}}.
\end{equation}
\end{theorem}
\begin{proof}
Fix $\epsilon > 0$ and let $\ket{\psi_0}_{AE}$ be the state the adversarial source Eve creates, sending the $A$ portion to Alice.  Using Theorem \ref{thm:sample} (with respect to the reference basis $\{W^T\ket{0,0}, \cdots, W^T\ket{1,P-1}\}$), there exist ideal states $\{\ket{\phi^t}\}$, indexed over all subsets $t\subset\{1, 2, \cdots, N\}$ of size $m$, such that $\ket{\phi^t} \in \text{span}(\ket{w_{i_1}w_{i_2}\cdots w_{i_N}} \st |\hw(i_t) - \hw(i_{-t})| \le \delta) \otimes\mathcal{H}_E$ and Equation \ref{eq:ideal} holds.
(Note we define $\ket{w_0} = \ket{w_{0,0}} = W^T\ket{0,0}$.) From Equation \ref{eq:err-cl}, by setting $\delta$ as in Equation \ref{eq:delta},
we have $\sqrt{\epsilon^{cl}_\delta} = \epsilon$.

We now use a two-step proof method we developed in \cite{krawec2019quantum,krawec2020new} to utilize quantum-sampling for entropic uncertainty relations.  Here, we modify the first step of the proof for this cryptographic application, while the second step remains largely the same.  The first step is to analyze the security of the ideal state $\sigma_{TAE} = \frac{1}{T}\sum_t\kb{t}\otimes\kb{\phi^t}$.  Choosing a subset is equivalent to measuring the $T$ register in $\sigma_{TAE}$ causing the state to collapse to the given ideal state $\ket{\phi^t}$.  At this point, a measurement using $\W$ is made on subset $t$ resulting in some outcome $q$.  The post-measurement state, discarding those systems that were measured, is easily seen to be of the form:
\[
\phi^t_q = \sum_{k\in\alphabet_{2P-1}^{wt(q)}}p_k\underbrace{P\left(\sum_{i\in J_q^{(k)}}\alpha_i\ket{w_i}\ket{E_i}\right)}_{\sigma_{AE}^{(k)}}.
\]
with $P(z) = zz^*$ and $J_q^{(k)} \subset \{i\in\mathcal{A}_P^n \st |w(i) - w(q)|\le\delta\}.$  Recall $n=N-m$.

Let us consider one of the $\sigma_{AE}^{(k)}$ states and perform a measurement using POVM $\Z$ on the remaining $A$ portion.  
To compute this state, we write a single quantum walker $\ket{w_i} \in \mathcal{H}_W$ as:
$  \ket{w_i} = \ket{0}\ket{\phi(0,i)} + \ket{1}\ket{\phi(1,i)},$
where $\ket{\phi(c,i)}$ are (not necessarily normalized) states in $\mathcal{H}_P$.  Using this notation, after some algebra, we find that the post-measurement state, with Alice storing the outcome $z\in\alphabet^n_P$ in a classical register $Z$ and also tracing out the unmeasured coin register is:
\[
\sigma_{ZE}^{(k)} = \sum_{z\in\alphabet_P^n} \kb{z}_Z\sum_{i,j \in J_q^{(k)}}\alpha_i\alpha_j^*\sum_{c\in\{0,1\}^n}x_{z,c,i}x_{z,c,j}^*\otimes\ket{E_i}\bra{E_j}
\]
where given a string $c\in\{0,1\}^n$, $z\in\alphabet_P^n$, and $i\in J_q^{(k)}$, we define $x_{c,z,i}$ as:
$x_{z,c,i} = \prod_\ell\braket{z_\ell|\phi(c_\ell|i_\ell)}.$
To compute the min-entropy of this state, we will consider the following density operator:
\[
\chi_{ZE} = \sum_z \kb{z}\sum_{i\in J_q^{(k)}}|\alpha_i|^2\sum_c|x_{c,z,i}|^2\otimes\kb{E_i}
\]
Using a proof similar to a lemma in \cite{renner2008security} which bounds the min-entropy of a superposition based on the min-entropy of a suitable mixed state, we find that:
\[
\Hmin(Z|E)_{\sigma^{(k)}} \ge \Hmin(Z|E)_\chi - \log|J_q^{(k)}|.
\]
Note that, though the lemma in \cite{renner2008security} is not immediately applicable to the above scenario, the proof is, indeed, identical and so we omit the details for space reasons.

Consider the state $\chi_{ZEI}$ where we append an auxiliary system spanned by orthonormal basis $\ket{i}$:
\[
\chi_{ZEI} = \sum_{i\in J_q^{(k)}}|\alpha_i|^2\underbrace{\left(\sum_z \kb{z}\sum_c|x_{c,z,i}|^2\right)}_{\chi^{(i)}}\otimes\kb{E_i}\otimes\kb{i}
\]
For strings $z \in \alphabet_P^n$ and $i \in \alphabet_{2P}^n$, let $p(z|w_i)$ be the probability that outcome $z$ is observed if measuring the pure, and unentangled state, state $\ket{w_{i_1}w_{i_2}\cdots w_{i_n}}$ using POVM $\Z$.  Simple algebra shows that this is in fact $p(z|w_i) = \sum_c|x_{c,z,i}|^2$.  Thus $\chi^{(i)} = \sum_z\kb{z}p(z|w_i)$.
From Equation \ref{eq:cl-ent} and treating the joint $EI$ register as a single classical register, we have:
\[
\Hmin(Z|E)_\chi \ge \Hmin(Z|EI)_\chi \ge \min_i\Hmin(Z)_{\chi^{(i)}}.
\]
Fix a particular $i \in J_q^{(k)}$ and let $\eta = n-wt(i)$ (namely, $\eta$ is the number of zeros in the string $i$). Then, it is clear that:
$p(z|w_i) \le \max_{x\in\alphabet_P} Pr_W(\ket{w_0} \rightarrow x)^\eta = \gamma^\eta,$
where $\gamma$ was defined in Equation \ref{eq:gamma}.  Indeed, any other $Pr_W(\ket{w_i}\rightarrow z) \le 1$ and so we may consider only the $\ket{w_0}$ terms as contributing to this upper-bound.  From this, it follows that:
$  \Hmin(Z)_{\chi^{(i)}} = -\log \max_z p(z|w_i)\ge -\log\gamma^{n-wt(i)}.$

Now, since $i \in J_q^{(k)}$, we know that $wt(i) \le n(w(q)+\delta)$ and so:
\begin{align*}
\Hmin(Z|E)_\chi &\ge \Hmin(Z|EI)_\chi \ge \min_i\Hmin(Z)_{\chi^{(i)}}\\
&\ge -\log\gamma^{n(1-w(q) - \delta)} = -\eta_q\log\gamma.
\end{align*}

Finally, we note that $|J_q^{(k)}| \le d^{n\Hextd_{2P}(w(q)+\delta)}$ (using the well known bound on the volume of a Hamming ball), we have:
\begin{align*}
&\Hmin(Z|E)_\sigma \ge \min_k \Hmin(Z|E)_{\sigma^{(k)}}\\
&\ge -\eta_q\log_2\gamma - n\cdot\frac{\Hextd_{2P}(w(q) + \delta)}{\log_{2P} (2)}
\end{align*}

Of course, this is the ideal state analysis.  However, we may use a similar technique that we employed in \cite{krawec2019quantum} for translating this ideal analysis to the real case.  Indeed, let $\rho_{ZE}^{t,q}$ be the state of the real system, $\ket{\psi}$, conditioned on the protocol sampling subset $t$ and observing outcome $q$ and let $\sigma_{ZE}^{t,q}$ be the same for the ideal state.  If we define:
$\Delta_{t,q} = \frac{1}{2}\trd{\rho_{ZE}^{t,q} - \sigma_{ZE}^{t,q}},$
then, treating $\Delta_{t,q}$ as a random variable over the choice of $t$ and outcome $q$, it can be shown (see the proof of Theorem 2 in \cite{krawec2019quantum} for explicit details) that except with probability $\epsilon^{1/3}$, it holds that $\Delta_{t,q} \le \epsilon + \epsilon^{1/3}$ where the probability is over the choice of $t$ and outcome $q$.  Thus, by switching to smooth min entropy, we have, except with probability at most $\epsilon^{1/3}$ that $\Hmin^{2\epsilon+2\epsilon^{1/3}}(Z|E)_\rho \ge \Hmin(Z|E)_\sigma$.
Privacy amplification (Equation \ref{eq:PA}, setting the right-hand-side of that equation equal to $\epsilon_{PA} = 5\epsilon+4\epsilon^{1/3}$, namely twice the smoothening parameter plus an additional $\epsilon$), along with the fact that it requires $\log{N \choose m},$ random bits to choose a subset of size $m$, completes the proof.
\end{proof}


$ $\newline\textbf{Evaluation: }
We evaluate the performance of our protocol under a variety of position dimensions $P$.  Ordinarily, users would run the protocol and observe $q$ directly; however, to simulate its execution, we will assume the noise follows a depolarization channel with parameter $Q$.  We do this only to evaluate our protocol; furthermore, this noise model is a standard one to evaluate on in simulations.  From this, after sampling, Alice will have an expected Hamming weight in her test measurement of $w(q) = Q$.  In our evaluations, we will set $\epsilon = 10^{-36}$ which will imply a failure probability, and an $\epsilon_{PA}$-secure string, both on the order of $10^{-12}$.  We also use a sample size that is the square-root of the total number of signals $N$, namely $m = \sqrt{N}$.  Finally, to evaluate our bit-generation rate, we will require $\gamma$ (Equation \ref{eq:gamma}).  Since the walk settings are chosen by the user, we wrote a program that, for fixed dimension $P$, found the minimum $\gamma$ value over all time settings $T = 1, 2, \cdots, 5000$.  The evaluation of the bit generation rate of this SI-QW-QRNG protocol, using our analysis in Theorem \ref{thm:main}, is shown in Figure \ref{fig:1}.  A comparison to an alternative SI-QRNG protocol from \cite{xu2016experimental} is shown in Figure \ref{fig:2}.  Note that as the dimension of the walker increases, the bit-generation rates, even under high noise levels, increases.  Interestingly, as shown in Figure \ref{fig:2}, depending on the walker dimension, the QW based protocol can sometimes outperform the SI-QRNG protocol from \cite{xu2016experimental} (which is based on mutually unbiased measurements of a highly entangled state).

\begin{figure}
  \centering
  \includegraphics[width=0.48\linewidth]{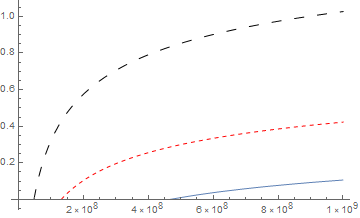}
  \includegraphics[width=0.48\linewidth]{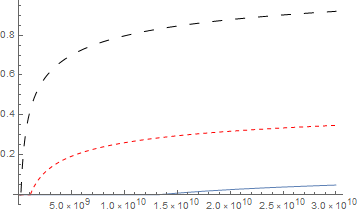}
\caption{Random bit generation rates of the QW-QRNG protocol.  $x$-axis: number of signals sent $N$ (from which $m=\sqrt{N}$ are used for sampling); $y$-axis: random bit-generation rate (namely $\ell/N$ where $\ell$ is computed using Theorem \ref{thm:main}).  Black dashed (top) is $P=51$; red-dashed (middle) is $P=11$; blue solid (lowest) is $P=5$.  Left graph is with $15\%$ noise in the source (namely $w(q) = 0.15$); Right graph has $20\%$ noise.}\label{fig:1}
\end{figure}

\begin{figure}
  \centering
  \includegraphics[width=0.48\linewidth]{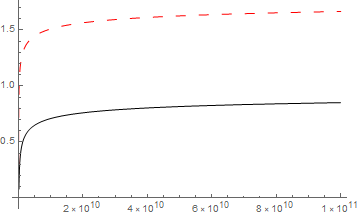}
  \includegraphics[width=0.48\linewidth]{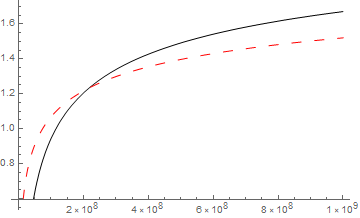}
\caption{Comparing the QW-QRNG protocol's bit generation rate (black-solid) with that of the SI-QRNG protocol in \cite{xu2016experimental} (red-dashed).  Left: $P=5$; Right: $P=51$.  In both cases we assume $10\%$ noise in the signal state.  For the SI-QRNG protocol's evaluation from \cite{xu2016experimental}, we use a dimension of $2P$.}\label{fig:2}
\end{figure}

\section{Closing Remarks}

In this paper, we modified, minimally, a QRNG protocol from \cite{QW-QRNG}, based on quantum walks, to be secure in the semi-source independent (SI) model.  Since standard entropic uncertainty relations cannot be directly applied as discussed, we develop an alternative entropic uncertainty relation for this protocol, showing it is secure in the SI model.  Our methods may potentially find applications in other difficult to analyze quantum cryptographic protocols.   There are important reasons for studying this QW-based protocol.  For instance, it is important to harness alternative quantum processes such as quantum-walk states, as it is still unclear what future experimental developments will yield; being able to utilize QW states may be highly relevant, especially since they are also useful for other tasks, computational and cryptographic, as discussed earlier.   Second, it is interesting from a theoretical stand-point.  Many exciting open problems remain, in particular a more rigorous evaluation of the performance of this QW-QRNG protocol for different walk parameters (such as alternative coin operators) or alternative models (such as history-dependent walks \cite{rohde2013quantum,mcgettrick2009one,krawec2015history,brun2003quantum}) would be very exciting.

\balance

\end{document}